\documentclass[10pt,a4paper,reqno]{amsart}
\usepackage{amssymb}
\usepackage{amsmath}
\usepackage{enumerate}
\usepackage{subfig}
\usepackage{graphicx}
\usepackage{tjmm}
\usepackage{multirow}

\usepackage[utf8x]{inputenc}

\newtheorem{property}[theorem]{Property}

\coordinates{11}{2019}{1-2}{77-90}

\title[Cubic spline approximation]{Cubic spline approximation of the reliability polynomials of two dual hammock networks}

\author[G. Cristescu]{Gabriela Cristescu}

\address[G. Cristescu]{Department of Mathematics and Computer Sciences,\newline\indent ''Aurel Vlaicu'' University of Arad, Bd. Revolu\c{t}iei, No. 77, 310130-Arad, Romania.}
\email{{\tt gabriela.cristescu@uav.ro}}

\author[V-F. Dr\v{a}goi]{Vlad-Florin Dr\v{a}goi}
\address[V-F. Dr\v{a}goi]{Department of Mathematics and Computer Sciences,\newline\indent "Aurel Vlaicu" University of Arad, Bd. Revolu\c{t}iei, No. 77, 310130-Arad, Romania, and\newline\indent
LITIS, University of Rouen Normandie, Avenue de l’universit\'{e}, 76801 Saint-\'{E}tienne-du-Rouvray, France.}
\email{\tt vlad.dragoi@uav.ro}

\subjclass[2010]{41A15, 41A29,68Q17}
\keywords{approximation, Bernstein basis, complementarity, dual networks, hammock network, reliability polynomial, spline function}

\begin{document}

\begin{abstract}
The property of preserving the convexity and concavity of the Bernstein polynomial and of the B\'{e}zier curves is used to generate a method of approximating the reliability polynomial of a hammock network. The mutual behaviour of the reliability polynomials of two dual hammock networks is used to generate a system of constraints since the initial information is not enough for using a classical approximation scheme. A cubic spline function is constructed to generate approximations of the coefficients of the two reliability polynomials. As consequence, an approximation algorithm is described and tested through simulations on hammocks with known reliability, comparing the results with the results of approximations attempts from literature.
\end{abstract}

\maketitle

\section{Introduction}

Estimating the reliability of a two-terminal network (2TN) is a long standing problem in network reliability. It started with the seminal work of Moore and Shannon \cite{MS1,MS2}. In their model the reliability of a two-terminal network ($N$) is defined as the $(s,t)$ connectedness of $N$. However, their methods for improving the reliability take into account the connectedness, as well as the non-connectedness, as $N$ is intended to work as a switch. Hence, they propose a particular family of networks, known as hammocks. Even though, they did not provide any proof of the fact that for a given length $l$ and width $w$, hammocks are the most reliable minimal two-terminal networks, this fact was verified for small values of $w$ and $l$ \cite{CBDP2}. Despite of its particular strong structure, there is no efficient algorithm for computing the reliability of a hammock network. In general, for a random two-terminal network the combinatorial and reliability problems are difficult \cite{B1,V1}.

There are two directions on which the research community followed this topic: firstly, decreasing the complexity of the state-of-the-art algorithms for the exact computation of the reliability polynomials, and secondly, estimating the coefficients of the reliability polynomial for more accurate approximations. Here, we will tackle the second direction and propose a new method for approximating the reliability polynomial of a hammock network, having a very low time complexity. Based on the latest results of Huh \cite{H} and Lenz \cite{L1} the sequence of the coefficients of a the reliability polynomial is log-concave. 
We propose an algorithm that uses the shape preserving properties of the Bernstein-B\'{e}zier type approximation operators in order to approximate the reliability polynomial in Section \ref{sec:approx}.

The structure of the paper is as follows. In Section \ref{sec:2} we introduce a concept of complementarity of two functions with respect to an operator. This is used all over the paper in order to describe the relationship between the reliability polynomials of two dual 2TNs, as presented in Section \ref{sec:3}. In Section \ref{sec:approx} we generate a mathematical model for simultaneously approximating the reliability polynomials of two dual 2TNs, in order to preserve the shape, in terms of shape of epigraph \cite{CL} following to the results from \cite{H} and \cite{L1}. An algorithm is elaborated for computing the approximant objects. Finally, numerical examples are presented, taking into account existing exact data about small size 2TNs \cite{CBDP2}. The results obtained using the new technique are compared with the approximations constructed by other authors \cite{BC1}.

\section{$\Delta$-complementary functions with respect to an operator}\label{sec:2}

Denote by $\mathbb{R}$ the set of real number and by $\mathbb{N}$ the set of natural numbers, $\mathbb{N^{*}}=\mathbb{N} \setminus \{0\}$. Suppose that $\mathfrak{F}$ is a class of real functions defined on some structured set and $\mathcal{O}:\mathfrak{F}\rightarrow \mathbb{R}$ is an operator.

\begin{definition}\label{Com}
Let $\Delta \in \mathbb{R}$. Two functions $f\in \mathfrak{F}$ and $g\in \mathfrak{F}$ are said to be $\Delta$-complementary with respect to operator $\mathcal{O}$ if
\begin{equation}\label{DCom}
\mathcal{O}(f) + \mathcal{O}(g) = \Delta.
\end{equation}
\end{definition}
\begin{remark}
It is obvious that the definition is consistent. Indeed, let us suppose that $\mathfrak{F}$ is a $n$-dimensional real linear space and $\{e_1, e_2, ..., e_n\}\subset \mathfrak{F}$ is a basis of this space. Suppose that $\mathcal{O}$ is a linear operator and $\mathcal{O}(e_1)=a \in \mathbb{R}$. Then functions $f\in \mathfrak{F}$ and $g\in \mathfrak{F}$ defined by
$$f=\alpha e_1,$$
$$g= \frac{\Delta-\alpha a}{a}e_1,$$
with $\alpha \in \mathbb{R}$, are $\Delta$-complementary with respect $\mathcal{O}$.
\end{remark}

\medskip

\section{Complementarity in hammocks}\label{sec:3}

Consider a hammock network $H_{(l,w)}$ of width $w\in \mathbb{N^{*}}$ and length $l\in \mathbb{N^{*}}$ and denote $n=lw$. The reliability polynomial of this network may be written, in Bernstein basis, $\{C_{n}^{k}p^{k}(1-p)^{n-k}| k\in \{0,1,...,n\}\}$ as:
\begin{equation}\label{h}
h_{(l,w)}(p)=\sum_{k=0}^{n}N_{k}p^{k}(1-p)^{n-k},
\end{equation}
where $p\in [0,1]$ and all coefficients $N_{k}=C_{n}^{k}a_k$  are non-negative real numbers, $a_k$ being the coefficients of $h_{(l,w)}$ in this basis, $k\in \{0,1,2,...,n\}$. Formula (\ref{h}) is known as the N-form of the reliability polynomial in the technical literature. Due to its frequent use in the literature, we will use the N-form in the sequel, referring by an abuse of language to $N_k$ as the coefficients of $h_{(l,w)}$ in Bernstein basis. More properties of the reliability polynomials are in \cite{MS1} and \cite{MS2}. The reliability polynomial is fully known in case of small dimension hammocks (see \cite{CBDP1}, \cite{CBDP2}). According to \cite{MS1}, the hammocks $H_{(l,w)}$ and $H_{(w,l)}$ are said to be dual hammocks. Suppose that
\begin{equation}\label{hd}
h_{(w,l)}(p)=\sum_{k=0}^{n}N_{k}^{\perp} p^{k}(1-p)^{n-k},
\end{equation}
with $N_{k}^{\perp}\geq 0$, $k\in \{0,1,2,...,n\}$ is the reliability polynomial of the dual hammock network $H_{(w,l)}$ of $H_{(l,w)}$.

It is proved in \cite{MS1}, formula (6) pp.197, that
\begin{equation}\label{hMS}
h_{(l,w)}(p)+h_{(w,l)}(1-p)=1.
\end{equation}

\begin{property}
Suppose that $H_{(l,w)}$ and $H_{(w,l)}$ are dual hammock networks and $h_{(l,w)}$ and $h_{(w,l)}$ defined by (\ref{h}) and (\ref{hd}) are their reliability polynomials, respectively. Then
\begin{equation}\label{su}
\sum_{k=0}^{n}N_{k}+\sum_{k=0}^{n}N_{k}^{\perp}= 2^{n}.
\end{equation}
\end{property}

\begin{proof} Since $$1=(1-p+p)^n =\sum_{k=0}^{n}C_{n}^{k} p^{k}(1-p)^{n-k},$$ $C_{n}^{k}$ denoting the combinations of $n$ elements taken by $k$, we have
\begin{equation}\label{h1}
h_{(l,w)}(p)+h_{(w,l)}(1-p)=\sum_{k=0}^{n} C_{n}^{k} p^{k}(1-p)^{n-k},
\end{equation}
for all $p\in [0,1]$. Taking $p=\frac{1}{2}$ in (\ref{h1}) one gets
$$\sum_{k=0}^{n}N_{k}\left(\frac{1}{2}\right)^{n}+\sum_{k=0}^{n}N_{k}^{\perp}\left(\frac{1}{2}\right)^{n}=\sum_{k=0}^{n} C_{n}^{k}\left(\frac{1}{2}\right)^{n}.$$
Dividing this equality by $\left(\frac{1}{2}\right)^{n}$ one gets (\ref{su}).
\end{proof}

\begin{remark}
Let us suppose that $\mathfrak{F}$ is the set of all polynomials $u:[0,1]\rightarrow [0,1]$ of degree at most $n=lw$, expressed in Bernstein basis, and operator $\mathcal{O}$ puts in correspondence each polynomial $u$ with the sum of its coefficients. Then (\ref{su}) means that the reliability polynomials of two dual hammocks are $2^n$-complementary functions with respect to the sum of their coefficients.
\end{remark}

\begin{remark}\label{FB}
Consider a hammock network $H_{(l,w)}$ and its dual $H_{(w,l)}$, with their reliability polynomials denoted as above. Let us define two functions
$F_{(l,w)}:[0,n]\rightarrow \mathbb{R}$ and $F_{(w,l)}:[0,n]\rightarrow \mathbb{R}$ by
\begin{equation}\label{F}
F_{(l,w)}(x)=
\begin{cases}
0, & \text{if $x=0$}\\
(N_k-N_{k-1})x+kN_{k-1}-(k-1)N_k, & \text{if $x\in [k-1,k], k\in\{1,2,...,n\}$}.
\end{cases}
\end{equation}
\begin{equation}\label{Fd}
F_{(w,l)}(x)=
\begin{cases}
0, & \text{if $x=0$}\\
(N_k^{\perp}-N_{k-1}^{\perp})x+kN_{k-1}^{\perp}-(k-1)N_k^{\perp}, & \text{if $x\in [k-1,k], k\in\{1,2,...,n\}$}.
\end{cases}
\end{equation}
We refer to function $F_{(l,w)}$ (respectively $F_{(w,l)}$) as the coefficients segmentary linear function of hammock network $H_{(l,w)}$ (respectively $H_{(w,l)}$). It is obvious that number $\int_{0}^{n}F_{(l,w)}(x)dx=\frac{1}{2}+\sum_{k=0}^{n} N_k$ (respectively $\int_{0}^{n}F_{(w,l)}(x)dx=\frac{1}{2}+\sum_{k=0}^{n} N_k^{\perp}$) equals to the area of the subgraph of function $F_{(l,w)}$ (respectively $F_{(w,l)}$) on $[0,n]$. We define, in the same manner, a function $B:[0,n]\rightarrow \mathbb{R}$ by
\begin{equation}\label{B}
B(x)=
\begin{cases}
1, & \text{if $x=0$}\\
(C_{n}^k-C_{n}^{k-1})x+kC_{n}^{k-1}-(k-1)C_{n}^k, & \text{if $x\in [k-1,k], k\in\{1,2,...,n\}$}.
\end{cases}
\end{equation}
The area of the subgraph of function $B$ on $[0,n]$ equals to $2^n$. Equation (\ref{su}) is equivalent to
\begin{equation}
\int_{0}^{n}F_{(l,w)}(x)dx +\int_{0}^{n}F_{(w,l)}(x)dx=2^n,
\end{equation}
which means that functions $F_{(l,w)}$ and $F_{(w,l)}$ are $2^n$-complementary with respect to the definite integration over $[0,n]$.
\end{remark}

\begin{property}
Suppose that $H_{(l,w)}$ and $H_{(w,l)}$ are dual hammock networks and $h_{(l,w)}$ and $h_{(w,l)}$ defined by (\ref{h}) and (\ref{hd}) are their reliability polynomials, respectively. Then
\begin{equation}\label{suCo}
N_{k}+N_{n-k}^{\perp} = C_{n}^k,
\end{equation}
for all $k\in \{0,1,2,...,n\}.$
\end{property}
\begin{proof} Consider the linear space of polynomials of degree at most $n$ defined on $[0,1]$ expressed with respect to Bernstein basis. A polynomial in this space is a linear combination of the elements of this basis. As known, two polynomials expressed in the same basis are identical if and only if they have the same coefficients. Since (\ref{h1}) may be written as
$$ h_{(l,w)}(p)+h_{(w,l)}(1-p)=\sum_{k=0}^{n} C_{n}^{k} p^{k}(1-p)^{n-k},$$
it means that
\begin{align*}
\sum_{k=0}^{n} C_{n}^{k} p^{k}(1-p)^{n-k} &= \sum_{k=0}^{n}N_{k}p^{k}(1-p)^{n-k}+ \sum_{k=0}^{n}N_{k}^{\perp} p^{n-k}(1-p)^{k}\\
																		&= \sum_{k=0}^{n}(N_{k}+N_{n-k}^{\perp})p^{k}(1-p)^{n-k}
\end{align*}
for all $p\in [0,1]$ and for all $k\in \{0,1,2,...,n\}$. This identity holds if and only if (\ref{suCo}) is valid.
\end{proof}

\medskip

\section{Simultaneous approximation of the reliability polynomials of two dual hammock networks}\label{sec:approx}

Consider a hammock network $H_{(l,w)}$ of width $w\in \mathbb{N^{*}}$ and length $l\in \mathbb{N^{*}}$ and denote $n=lw$. The reliability polynomial of this network, expressed in Bernstein basis, is $h_{(l,w)}$ defined by (\ref{h}). Knowing the reliability polynomial $h_{(l,w)}$ is equivalent to knowing the corresponding function $F_{(l,w)}$ defined by (\ref{F}). We consider the dual network $H_{(w,l)}$, together with the corresponding functions defined above by (\ref{hd}) and (\ref{Fd}). In his section we intend to build a method of approximation of functions $F_{(l,w)}$ and $F_{(w,l)}$ by means of a cubic spline function, starting from the properties of the reliability polynomials described in \cite{CBDP1} and \cite{CBDP2}. Some generalized convexity properties as described in \cite{CL} will be used. We construct segmentary cubic polynomials meant to imitate the shape of functions $F_{(l,w)}$ and $F_{(w,l)}$. As proved in \cite{TP}, given a continuous function on a bounded closed interval, the Bernstein approximation polynomial of degree $s$ of this function preserves the convexity of the approximated function (see also \cite{L2} and \cite{N}). This property gave us the idea of approximating functions $F_{(l,w)}$ and $F_{(w,l)}$ by means of polynomials imitating the Bernstein polynomial of third degree. As known, the Bernstein approximation polynomial of degree $s$ of a continuous function is defined by using the values of the approximated function on $s+1$ equidistant knots, including the extremities of the interval taken into account. Our attempt of approximating functions $F_{(l,w)}$ and $F_{(w,l)}$ do not benefit of information of this kind. The framework of the approximation described in this section is given by incomplete data and non-equidistant knots.

\subsection{Mathematical model}

The known data on the coefficients of $h_{(l,w)}$ and $h_{(w,l)}$, written in terms of functions $F_{(l,w)}$ and $F_{(w,l)}$, are:

\begin{equation}
\left\{
  \begin{array}{l @{} l  @{} l}
    F_{(l,w)}(k)&=0, &\forall k\in \{0, 1,...,l-1\}\\
	F_{(l,w)}(l)&=N_{l} >0& \\
    F_{(l,w)}(l+t)&=N_{l+t} > N_{l} &\\
    F_{(l,w)}(k)&=C_{n}^{k},& \forall k\in \{n-w+1, n-w+2,...,n\}
  \end{array}
\right.
\end{equation}
where $t$ is some fixed number belonging to $\{l+1, l+2, ..., n-w\}.$
\begin{equation}
\left\{
  \begin{array}{l@{} l  @{} l}
    F_{(w,l)}(0)=0, &\forall k\in \{0, 1,...,w-1\}\\
    F_{(w,l)}(w)=N_{w}^{\perp} > 0& \\
    F_{(w,l)}(w+s)=N_{w+s}^{\perp} > N_{w}^{\perp} &\\
    F_{(w,l)}(k)=C_{n}^{k}, & \forall k\in \{n-l+1, n-l+2,...,n\}
  \end{array}
\right.
\end{equation}
with $s$ having some fixed value from $\{w+1, w+2, ..., n-l\}$, $s\neq t$, $s\neq n-t.$
Values $N_{l}$, $N_{w}^{\perp}$, $N_{l+t}$ and $N_{w+s}^{\perp}$, are known from \cite{CBDP1}, \cite{CBDP2} or may be computed or measured by means of other techniques.

In order to write the Bernstein approximation polynomial of degree $s$ of a function on $[a,b]$, the values of this function on $s+1$ equidistant points from this interval, including the extremities, are needed. As one can see, the above data do not provide sufficient information from the perspective of defining an enough refined division of interval $[0,n]$ in order to approximate function $F_{(l,w)}$ and $F_{(w,l)}$ with a convenient error, according to the known results on the degree of approximation from the mathematical literature. As consequence, we need to use more information for building a satisfying approximant of $F_{(l,w)}$ and $F_{(w,l)}$. We prefer to generate approximations of these two functions by imitating the Bernstein-type approximation in this case of incomplete data. We replace the missing knowledge on the values of the approximated function on equidistant intermediary points by some conditions on some bridge points, constructed according to (\ref{su}) and (\ref{suCo}). We call this process of approximation a \emph{pseudo-Bernstein type approximation}. \\
In order to approximate $F_{(l,w)}$ and $F_{(w,l)}$ we construct two continuous cubic spline functions $f_{(l,w)}:[0,n]\rightarrow \mathbb{R}$ and $f_{(w,l)}:[0,n]\rightarrow \mathbb{R}$ that verify the following conditions:

\begin{equation}
\left\{
  \begin{array}{ll}
    f_{(l,w)}(k)=0, &k\in\{0,\dots,l-1\}\\
    f_{(l,w)}(l)=N_{l}& \\
    f_{(l,w)}(n-w+k)=C_{n}^{w-k},& k\in \{0,\dots,w-1\},
  \end{array}
\right.
\end{equation}

\begin{equation}
\left\{
  \begin{array}{ll}
    f_{(w,l)}(k)=0,&k\in\{0,\dots,w-1\} \\
    f_{(w,l)}(w)=N_{w}^{\perp} &\\
    f_{(w,l)}(n-l+k)=C_{n}^{l-k},& k\in \{0,\dots,l-1\}.
  \end{array}
\right.
\end{equation}
and the connecting conditions resulting from the duality and complementarity properties discussed in the previous section:
\begin{equation}
\left\{
  \begin{array}{ll}
    f_{(l,w)}(x_1)+f_{(w,l)}(n-x_1)=C_n^{x_1}, & \text{and either} \\
    f_{(l,w)}(x_2)+f_{(w,l)}(n-x_2)=C_n^{x_2}, & \text{or} \\
    \int_{0}^{n}f_{(l,w)}(x)dx +\int_{0}^{n}f_{(w,l)}(x)dx=2^n, & \text{or}\\
    \sum_{k=0}^{n}f_{(l,w)}(k)+\sum_{k=0}^{n}f_{(w,l)}(k)=2^n.
  \end{array}
\right.
\end{equation}
Here $x_1$ and $x_2$ are two natural numbers taken from interval $(\max\{l+t,w+s\}, \min\{n-w+s,n-l+t\})$.\\
In order to define the two functions $f_{(l,w)}$ and $f_{(w,l)}$ we take into account the properties of operators defined by S.N. Bernstein \cite{Ber}, T. Popoviciu \cite{TP}, P. B\'{e}zier \cite{Bez1} \cite{Bez2} for approximation of continuous functions and of plane curves. These types of polynomial operators preserve some shape properties of the approximated curve, as superior order convexity and concavity (see \cite{TP}). Function $f_{(l,w)}$ is searched as:
\begin{equation}\label{flw}
f_{(l,w)}(x)=
\begin{cases}
0, & \text{if $0\leq x\leq l-1$}\\
N_{l} x+N_{l}(1-l), & \text{if $l-1< x\leq l$}\\
\tilde{B}_{(l,w)}(x), & \text{if $l<x\leq n-w$}\\
d_{(l,w)}(k)(x), & \text{if $x\in (k-1, k], k\in\{n-w+1,...,n\}$}
\end{cases}
\end{equation}
Here
\begin{align*}
d_{(l,w)}(n-w+1)(x)&=(C_{n}^{w-1}-N_{n-w})x+N_{n-w}(n-w+1)-C_{n}^{w-1}(n-w),\\
d_{(l,w)}(k)(x)&=(C_{n}^{k}-C_{n}^{k-1})x+kC_n^{k-1} - (k-1)C_{n}^{k}
\end{align*}
are the straight line segments determined by points $(k-1, C_n^{k-1})$ and $(k, C_n^{k})$, for all $k\in\{n-w+2,...,n\}$ respectively. Also, $\tilde{B}(l,w)$ denotes the cubic pseudo-Bernstein type approximation polynomial,
\begin{align*}
\tilde{B}_{(l,w)}(x)&=\frac{1}{(n-w-l)^3}\left[N_{l} (n-w-x)^3 +aC_3^{1}(n-w-x)^2 (x-l)\right]\\
&+\frac{1}{(n-w-l)^3}\left[bC_3^{2}(n-w-x)(x-l)^2 + N_{n-w}(x-l)^3\right],
\end{align*}
that takes the values $N_{l}$ and $N_{n-w}$ at the extremities of interval $(l, n-w].$
Function $f_{(w,l)}$ is searched as:
\begin{equation}\label{fwl}
f_{(w,l)}(x)=
\begin{cases}
0, & \text{if $0\leq x\leq w-1$}\\
N_{w}^{\perp} x+N_{w}^{\perp}(1-w), & \text{if $w-1< x\leq w$}\\
\tilde{B}_{(w,l)}(x), & \text{if $w<x\leq n-l$}\\
d_{(w,l)}(k)(x), & \text{if $x\in (k-1, k], k\in\{n-l+1,...,n\}$}
\end{cases}
\end{equation}
Here
\begin{align*}
d_{(w,l)}(n-l+1)(x)&=(C_{n}^{l-1}-N_{n-l}^{\perp})x+N_{n-l}^{\perp}(n-l+1)-C_{n}^{l-1}(n-l),\\
d_{(w,l)}(k)&=(C_{n}^{k}-C_{n}^{k-1})x+kC_n^{k-1} - (k-1)C_{n}^{k}
\end{align*}
are the straight line segments determined by points $(k-1, C_n^{k-1})$ and $(k, C_n^{k})$, for all $k\in\{n-l+2,...,n\}$. As above,
\begin{align*}
\tilde{B}_{(w,l)}(x)&=\frac{1}{(n-w-l)^3}\left[N_{w}^{\perp}(n-l-x)^3 +cC_3^{1}(n-l-x)^2 (x-w)\right]\\
&+\frac{1}{(n-w-l)^3}\left[dC_3^{2}(n-l-x)(x-w)^2 + N_{n-l}^{\perp}(x-w)^3\right],
\end{align*}
takes the values $N_{w}^{\perp}$ and $N_{n-l}^{\perp}$ at the extremities of interval $(w, n-l].$
The values $N_{l}$, $N_{l+t}$, $N_{w}^{\perp}$ and $N_{w+s}^{\perp}$ are computed using \cite{BC1}, and the values $N_{n-w}$ and $N_{n-l}^{\perp}$ are obtained by means of (\ref{suCo}).
The system of equations produced by using the above mentioned conditions in order to compute the coefficients $a, b, c, d$ of functions $f_{(l,w)}$ and $f_{(w,l)}$ is:
\begin{equation}
\left\{
  \begin{array}{ll}
    f_{(l,w)}(l+t)=N_{l+t} \\
    f_{(w,l)}(w+s)=N_{w+s}^{\perp} \\
    f_{(l,w)}(x_1)+f_{(w,l)}(n-x_{1})=C_{n}^{x_{1}}, & \text{and either} \\
    f_{(l,w)}(x_2)+f_{(w,l)}(n-x_2)=C_n^{x_2}, & \text{or} \\
    \int_{0}^{n}f_{(l,w)}(x)dx +\int_{0}^{n}f_{(w,l)}(x)dx=2^n, & \text{or}\\
    \sum_{k=0}^{n}f_{(l,w)}(k)+\sum_{k=0}^{n}f_{(w,l)}(k)=2^n.
  \end{array}
\right.
\end{equation}

The best results in applications were obtained taking $x_2 =n-x_1$ in the fourth equation. The third and fourth equations are consequences of both the fifth and the sixth equations as proved by (\ref{suCo}). Hence, we use the system of equations consisting in the first four equations. Let us introduce the following notations in order to write the detailed form of the above system of equations:
$$p_{(l,w)}^{(s;k)}(x)=(x-l)^{k}(n-w-x)^{s-k},$$
$$p_{(w,l)}^{(s;k)}(x)=(x-w)^{k}(n-l-x)^{s-k}.$$
As consequence, the above system of equations becomes:
\begin{equation}\label{system}
\left\{
  \begin{array}{ll}
   3(n-w-l-t)^{2}ta+ 3(n-w-l-t)t^{2}b=A_1 \\
   3(n-w-l-s)^{2}sc+ 3(n-w-l-s)s^{2}d=A_2 \\
   3p_{(l,w)}^{(3;1)}(x_1)a+3p_{(l,w)}^{(3;2)}(x_1)b+3p_{(w,l)}^{(3;1)}(n-x_1)c+3p_{(w,l)}^{(3;2)}(n-x_1)d=A_3 \\
   3p_{(l,w)}^{(3;1)}(x_2)a+3p_{(l,w)}^{(3;2)}(x_2)b+3p_{(w,l)}^{(3;1)}(n-x_2)c+3p_{(w,l)}^{(3;2)}(n-x_2)d=A_4,
  \end{array}
\right.
\end{equation}
with the following notations to compute the right side of each equation:
\begin{align*}
A_{1}&=N_{l+t}(n-w-l)^{3}-N_{l}(n-w-l-t)^{3}-N_{n-w}t^{3},\\
A_2&=N_{w+s}^{\perp}(n-w-l)^{3}-N_{w}^{\perp}(n-w-l-s)^{3}-N_{n-l}^{\perp}s^{3},\\
A_3&=C_{n}^{x_1}(n-w-l)^3-C_{n}^{l}(n-w-x_{1})^{3}-C_{n}^{w}(x_1-l)^3,\\
A_4&=C_{n}^{x_2}(n-w-l)^3-C_{n}^{l}(n-w-x_{2})^{3}-C_{n}^{w}(x_2-l)^3.
\end{align*}

\subsection{Estimation of the error}
The result in this subsection is a rough estimation of the error of the approximation from the algorithm described above. It may be improved by using additional information on the input points $x_1,x_2.$
\begin{theorem}\label{thm:error-approx}
Let us denote by $\tilde{h}_{(l;w)}$ and $\tilde{h}_{(w;l)}$ the approximation of the reliability polynomials $h_{(l;w)}$, respectively $h_{(w;l)}$ of two dual hammocks, obtained by means of the above described approximation procedure. Let us denote by
$$M=\max\{(l+1)^{l+1}(n-l-1)^{n-l-1};(w+1)^{w+1}(n-w-1)^{n-w-1}\}.$$
Then the error and the cumulative error of the simultaneous approximation of the reliability polynomials of the two dual hammocks is estimated as:
\begin{align}
\left|h_{(l,w)}(p)-\tilde{h}_{(l,w)}(p)\right|&\leq \frac{M(n-w-l-1)}{n^n}\left|C_n^{\left[\frac{n}{2}\right]}-\min(C_n^{l+1};C_n^{w+1})\right|,\label{eq:21}\\
\left|h_{(w,l)}(p)-\tilde{h}_{(w,l)}(p)\right|&\leq \frac{M(n-w-l-1)}{n^n}\left|C_n^{\left[\frac{n}{2}\right]}-\min(C_n^{l+1};C_n^{w+1})\right|,\label{eq:22}\\
\left|1-\tilde{h}_{(l,w)}(p)-\tilde{h}_{(w,l)}(p)\right|&\leq
\frac{2M(n-w-l-1)}{n^n}\left|C_n^{\left[\frac{n}{2}\right]}-\min(C_n^{l+1};C_n^{w+1})\right|,\label{eq:23}
\end{align}
for all $p\in[0,1].$
\end{theorem}
\begin{proof}
According to the previously described procedure,
$$\tilde{h}_{(l,w)}(p)=\sum_{k=0}^{n}f_{(l,w)}(k)p^{k}(1-p)^{n-k},$$
$$\tilde{h}_{(w,l)}(p)=\sum_{k=0}^{n}f_{(w,l)}(k)p^{k}(1-p)^{n-k}.$$
Then, according to (\ref{h}) and (\ref{hd}), in view of (\ref{hMS}) and (\ref{suCo}) one can compute:
\begin{align*}
1-\tilde{h}_{(l,w)}(p)-\tilde{h}_{(w,l)}(1-p)&=\sum_{k=0}^{n}\left[C_{n}^{k}-f_{(l,w)}\right]p^{k}(1-p)^{n-k}-\sum_{k=0}^{n}f_{(w,l)}(k)p^{n-k}(1-p)^{k}\\
&=\sum_{k=0}^{n}\left[C_{n}^{k}-f_{(l,w)}(k)-f_{(w,l)}(n-k)\right]p^{k}(1-p)^{n-k}\\
&=\sum_{k=0}^{n}\left[N_{k}-f_{(l,w)}(k)\right]p^{k}(1-p)^{n-k}\\
&+\sum_{k=0}^{n}\left[N_{n-k}^{\perp}-f_{(w,l)}(n-k)\right]p^{k}(1-p)^{n-k}.
\end{align*}
On another hand, functions $u_{k}:[0,1]\rightarrow \mathbb{R}$, $u_{k}(p)=p^{k}(1-p)^{n-k}$, $k\in\{1,2,...,n\}$, have the property that
$$\max\{u_{k}(p)|p\in [0,1]\}=u_{k}\left(\frac{k}{n}\right)=\frac{k^{k}(n-k)^{n-k}}{n^{n}},$$
which motivates the following majorant:
\begin{align*}
\left|h_{(l,w)}(p)-\tilde{h}_{(l,w)}(p)\right|&=\left|\sum_{k=l+1}^{n-w-1}\left[N_{k}-f_{(l,w)}(k)\right]p^{k}(1-p)^{n-k}\right|\\
&\leq \left|\sum_{k=l+1}^{n-w-1}\left[N_{k}-f_{(l,w)}(k)\right]\right|\frac{k^{k}(n-k)^{n-k}}{n^{n}}\\
&\leq \sum_{k=l+1}^{n-w-1}\left|\min(C_n^{l+1};C_n^{w+1})-C_n^{\left[\frac{n}{2}\right]}\right|\frac{k^{k}(n-k)^{n-k}}{n^{n}}.
\end{align*}
Also, one gets in the same manner:
$$\left|h_{(w,l)}(p)-\tilde{h}_{(w,l)}(p)\right|\leq \sum_{k=w+1}^{n-l-1}\left|\min(C_n^{l+1};C_n^{w+1})-C_n^{\left[\frac{n}{2}\right]}\right|\frac{k^{k}(n-k)^{n-k}}{n^{n}}.$$
Now, elementary computation shows that function $v:[0,n]\rightarrow \mathbb{R}$, $v(k)=k^{k}(n-k)^{n-k}$, has the property:
$$\min\{v(k)|k\in [0,n]\}=v\left(\frac{n}{2}\right)=\left(\frac{n}{2}\right)^n,$$
and its maximal value is taken at the frontier of the interval considered within its definition domain. Taking into account all these estimations and extremum properties, one gets the following upper bounds:
$$\left|h_{(l,w)}(p)-\tilde{h}_{(l,w)}(p)\right|\leq \frac{M(n-w-l-1)}{n^n}\left|C_n^{\left[\frac{n}{2}\right]}-\min(C_n^{l+1};C_n^{w+1})\right|,$$
$$\left|h_{(w,l)}(p)-\tilde{h}_{(w,l)}(p)\right|\leq \frac{M(n-w-l-1)}{n^n}\left|C_n^{\left[\frac{n}{2}\right]}-\min(C_n^{l+1};C_n^{w+1})\right|,$$
$$\left|1-\tilde{h}_{(l,w)}(p)-\tilde{h}_{(w,l)}(p)\right|$$
$$\leq\left|C_n^{\left[\frac{n}{2}\right]}-\min(C_n^{l+1};C_n^{w+1})\right|\left[\sum_{k=l+1}^{n-w-1}\frac{k^{k}(n-k)^{n-k}}{n^{n}}+\sum_{k=w+1}^{n-l-1}\frac{k^{k}(n-k)^{n-k}}{n^{n}}\right]$$
$$\leq \frac{2}{n^{n}}\left|C_n^{\left[\frac{n}{2}\right]}-\min(C_n^{l+1};C_n^{w+1})\right|\sum_{k=l+1}^{n-w-1}M$$
$$=\frac{2M(n-w-l-1)}{n^n}\left|C_n^{\left[\frac{n}{2}\right]}-\min(C_n^{l+1};C_n^{w+1})\right|,$$
as required.
\end{proof}
\begin{corollary} One can immediately deduce that \eqref{eq:21}, \eqref{eq:22}, \eqref{eq:23} give the estimation of the error in Chebyshev norm.
\end{corollary}
\subsection{The algorithm}\label{alg}
\begin{description}
  \item[Step 1] Compute the values $N_{l}$, $N_{l+1}$, $N_{w}^{\perp}$ and $N_{w+1}^{\perp}$, using the technique from \cite{BC1}. Compute the values $N_{n-w}$, $N_{n-w-1}$, $N_{n-l}^{\perp}$ and $N_{n-l-1}^{\perp}$ using (\ref{suCo}).
  \item[Step 2] Write the system of equations (\ref{system}) and solve it.
  \item[Step 3] Write functions $f_{(l,w)}$ and $f_{(w,l)}$ using (\ref{flw}) and (\ref{fwl}) respectively.
  \item[Step 4] Compute $f_{(l,w)}(k)$, $k\in \{l+2, l+3,..., n-w-2\}$ and write the approximate reliability polynomial $$\tilde{h}_{(l,w)}(p)=\sum_{k=0}^{n}f_{(l,w)}(k)p^{k}(1-p)^{n-k}.$$
  \item[Step 5] Compute $f_{(w,l)}(k)$, $k\in \{w+2, w+3,..., n-l-2\}$ and write the approximate reliability polynomial $$\tilde{h}_{(w,l)}(p)=\sum_{k=0}^{n}f_{(w,l)}(k)p^{k}(1-p)^{n-k}.$$
\end{description}

Eventually, we can add an extra step for evaluating the upper bound and the lower bound of the error of our approximation, according to Theorem \ref{thm:error-approx}.

\medskip

\section{Applications in case of unique solution}
\subsection{General model with unique solution}
System (\ref{system}) has unique solution if one takes $x_1=l+1$ and $x_2=n-w-1$. The conditions that generate the system of equations become, in view of (\ref{suCo}):
\begin{equation}
\left\{
  \begin{array}{ll}
    f_{(l,w)}(l+t)=N_{l+t} \\
    f_{(w,l)}(w+s)=N_{w+s}^{\perp} \\
    f_{(w,l)}(n-l-t)=C_{n}^{l+t}-N_{l+t}=N_{n-l-t}^{\perp} \\
    f_{(l,w)}(n-w-s)=C_{n}^{w+s}-N_{w+s}^{\perp}=N_{n-w-s} \\
    \end{array}
\right.
\end{equation}
After the required computation, the system of equations (\ref{system}) becomes:
\begin{equation}\label{unic}
\left\{
  \begin{array}{ll}
   3(n-w-t)^{2}ta+ 3(n-w-l-t)t^{2}b=A_1 \\
   3(n-w-s)^{2}sc+ 3(n-w-l-s)s^{2}d=A_2 \\
   3t^{2}(n-w-l-t)c+3t(n-w-l-t)^{2}d=A_3^1 \\
   3s^{2}(n-w-l-s)a+3s(n-w-l-s)^{2}b=A_4^1,
  \end{array}
\right.
\end{equation}
with the following notations to compute the right side of each equation:
\begin{align*}
A_1&=N_{l+t}(n-w-l)^{3}-N_{l}(n-w-l-t)^{3}-N_{n-w}t^{3},\\
A_2&=N_{w+s}^{\perp}(n-w-l)^{3}-N_{w}^{\perp}(n-w-l-s)^{3}-N_{n-l}^{\perp}s^{3},\\
A_3^1&=N_{n-l-t}^{\perp}(n-w-l)^3-N_{n-l}^{\perp}(n-w-l-t)^{3}-N_{w}^{\perp}t^3,\\
A_4^1&=N_{n-w-s}(n-w-l)^3-N_{n-w}(n-w-l-s)^{3}-N_{l}s^{3}.
\end{align*}

In fact, the system of equation (\ref{unic}) consists in two linear systems of two equations of two variables, as follows:
\begin{equation}\label{uni1}
\left\{
  \begin{array}{ll}
   3(n-w-t)^{2}ta+ 3(n-w-l-t)t^{2}b=A_1 \\
   3s^{2}(n-w-l-s)a+3s(n-w-l-s)^{2}b=A_4^1,
  \end{array}
\right.
\end{equation}

\begin{equation}\label{uni2}
\left\{
  \begin{array}{ll}
   3(n-w-s)^{2}sc+ 3(n-w-l-s)s^{2}d=A_2 \\
   3t^{2}(n-w-l-t)c+3t(n-w-l-t)^{2}d=A_3^1.
  \end{array}
\right.
\end{equation}

The algorithm \ref{alg} is now adapted to equations \eqref{unic}.

We use our algorithm to approximate the coefficients of some small dimension hammocks, i.e., the $3\times5$, and the $5\times 5$ hammocks. The implementation of the algorithm was done in Maple software. The exact reliability polynomials considered here are taken from \cite{CBDP2}. We also compute upper and lower bounds for such networks, more exactly Stanley type of bounds (see \cite{BC1}), denoted by

\begin{align*}
\mathsf{LB}&=\left[N_l, \left\{N_{l+1}\dfrac{C_{n}^{i}}{C_{n}^{l+1}}\right\}_{i=l+1\dots n-w-2},N_{n-w-1},N_{n-w},\left\{C_{n}^{i}\right\}_{i=n-w-1\dots n}\right];\\
\mathsf{UB}&=\left[N_l, N_{l+1},\left\{N_{n-w-1}\dfrac{C_{n}^{i}}{C_{n}^{n-w-1}}\right\}_{i=l+2\dots n-w-2},N_{n-w},\left\{C_{n}^{i}\right\}_{i=n-w-1\dots n}\right].\\
\end{align*}

Straightforward we obtain

\begin{align*}h_{(l,w)}(p)&\ge N_lp^l(1-p)^{n-l}+\sum\limits_{i=l+1}^{n-w-2}N_{l+1}\dfrac{C_{n}^{i}}{C_{n}^{l+1}}p^i(1-p)^{n-i}+N_{n-w-1}p^{n-w-1}(1-p)^{w+1}\\
&+N_{n-w}p^{n-w}(1-p)^{w}+\sum\limits_{i=n-w-1}^{n}C_{n}^{i}p^i(1-p)^{n-i};\\
h_{(l,w)}(p)&\le N_lp^l(1-p)^{n-l}+N_{l+1}p^{l+1}(1-p)^{n-l-1}+\sum\limits_{i=l+2}^{n-w-1}N_{n-w-1}\dfrac{C_{n}^{i}}{C_{n}^{n-w-1}}p^i(1-p)^{n-i}\\
&+N_{n-w}p^{n-w}(1-p)^{w}+\sum\limits_{i=n-w-1}^{n}C_{n}^{i}p^i(1-p)^{n-i}.
\end{align*}
\medskip

\subsection{The 3 by 5 hammock and its dual}~

\smallskip

For the 3 by 5 hammock we have the following results:
$$
\begin{array}{|r||r|r|r|r|r|r|r|r|}
\hline
\mathsf{LB}&21&194&249&249&194&116&1187&439\\
\hline
f_{(l,w)}(k)& 21&194&561&982&1320&1434&1187&439\\
\hline
N_k&21&194 &782&1772&2443&2114&1187&439\\
\hline
\mathsf{UB}&21&194&5596&5596&4352&2611&1187&439\\
\hline
\end{array}
$$
\begin{figure}
\centering
  \subfloat[For the 3 by 5 hammock.]{\includegraphics[width=.45\textwidth]{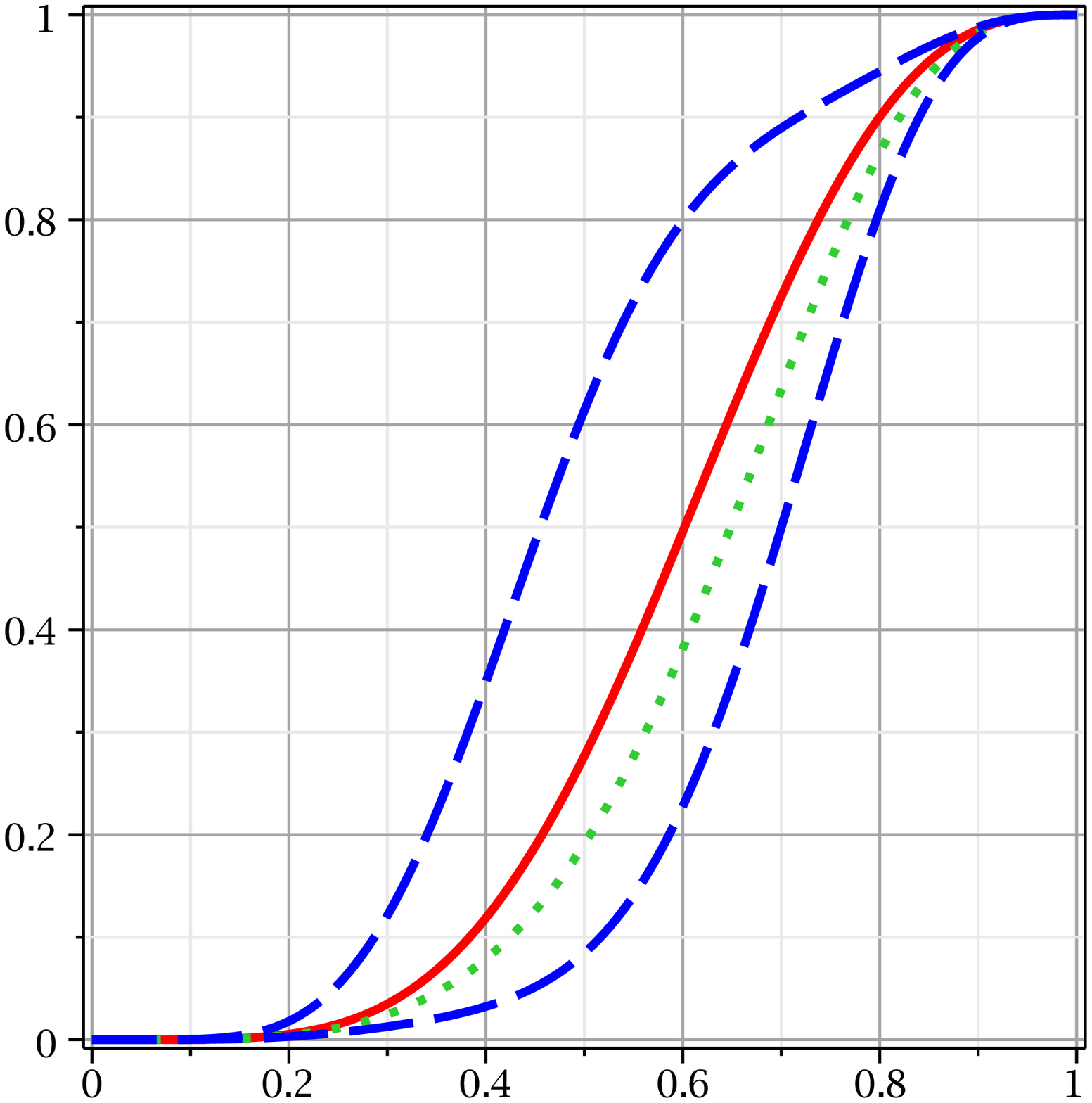}\label{fig1}}
  \hspace{0.1cm}
  \subfloat[For the dual (5 by 3) network.]{\includegraphics[width=.45\textwidth]{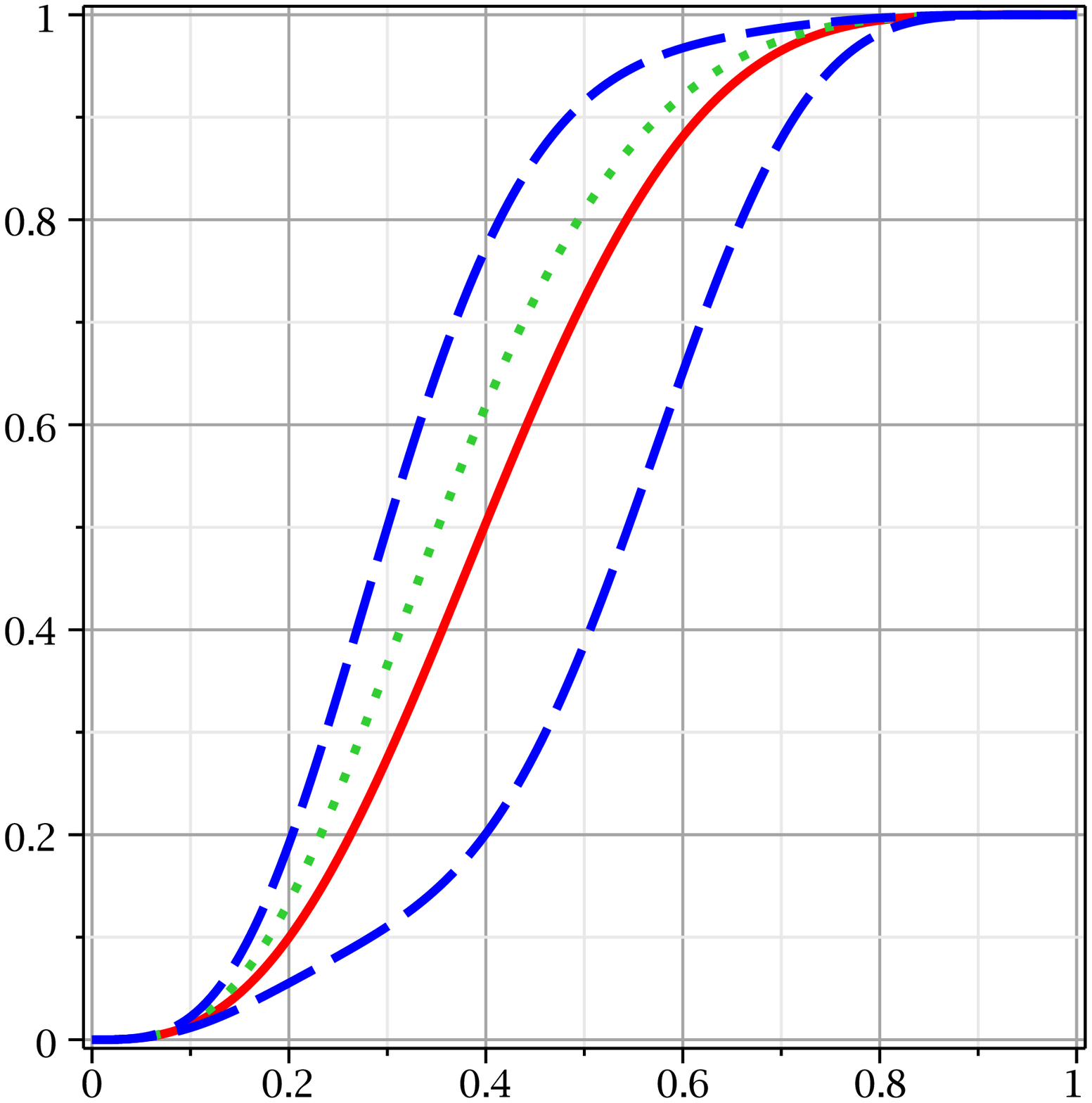}\label{fig2}}
  \vspace{0.1cm}
   \caption{Reliability polynomial (red solid line), our approximation of the polynomial with $s=1,t=1$ (green dotted line), the upper and the lower bounds (blue dashed lines).}\label{fig:1-2}
\end{figure}

As one can see, the approximation is satisfactory if the error is computed in Chebyshev norm or equivalence, fact illustrated in Fig. \ref{fig1} for the 3 by 5 hammock, and Fig. \ref{fig2} for its dual.

\vspace{2cm}

\subsection{The 5 by 5 hammock}~~

\vspace{1cm}

\begin{table}[!h]
\begin{center}
\begin{tabular}{|c|c||c||c||c||c|}
\hline
$k$&$\mathsf{LB}$&$N_k$&$\mathsf{UB}$&$f_{(l,w)}(k)$&$C_{n}^{k}$\\
\hline\hline
5&52&52&52&52&53130\\
6&994&994&994&994&177100\\
7&2698&8983&478002&20757&480700\\
8&6070&50796&1075504&55084&1081575\\
9&11466&200559&2031508&99716&2042975\\
10&18346&584302&3250414&150396&3268760\\
11&25018&1294750&4432382&202866&4457400\\
12&29187&2220298&5171113&252867&5200300\\
13&29187&2980002&5171113&296143&5200300\\
14&25018&3162650&4432382&328434&4457400\\
15&18346&2684458&3250414&345484&3268760\\
16&11466&1842416&2031508&343034&2042975\\
17&6070&1030779&1075504&316826&1081575\\
18&2698&471717&478002&262603&480700\\
19&176106&176106&176106&176106&177100\\
20&53078&53078&53078&53078&53130\\
\hline
\end{tabular}
\end{center}
\caption{Sequence of coefficients $N_k$ of the $5\times5$ hammock, the lower and upper bounds, the values of $f_{l,w}(k)$ for $s=t=1$, and the binomial coefficients, for $k\in\{l,\dots n-w\}.$}
\end{table}

\begin{figure}[!ht]
\includegraphics[width=0.5\textwidth,height=0.4\textwidth]{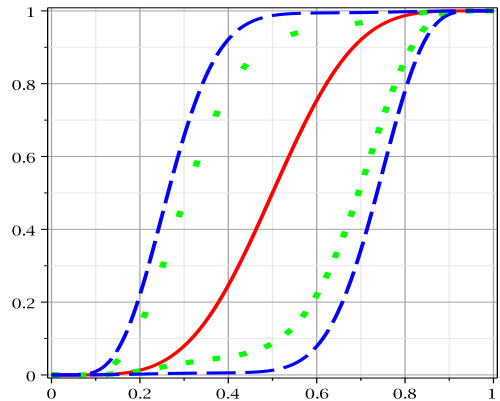}
\caption{Reliability polynomial of the 5 by 5 hammock networks (red solid line), our approximation of the polynomial with $s=1,t=1$ (green dotted line), the dual of our approximation (green space-dotted line), and the upper and the lower bounds (blue dashed lines).}
\end{figure}

\subsection{Improvements for the 5 by 5 hammock}~~

A natural improvement of our method can be applied when extra coefficients are known, fact that we illustrate in Fig. \ref{fig:spline+cubic}. Notice that with extra 3 coefficients (the green dotted line), i.e., the spline approximation using 7 points out of 16, gives extremely sharp results. However, we notice that our method applied with $s=9,t=1$ (magenta dash-dotted line in Fig. \ref{fig:spline+cubic}) provides better approximations for the last coefficients, than the Spline approximation with 7 points.

\begin{figure}[!ht]
\includegraphics[width=0.5\textwidth,height=0.4\textwidth]{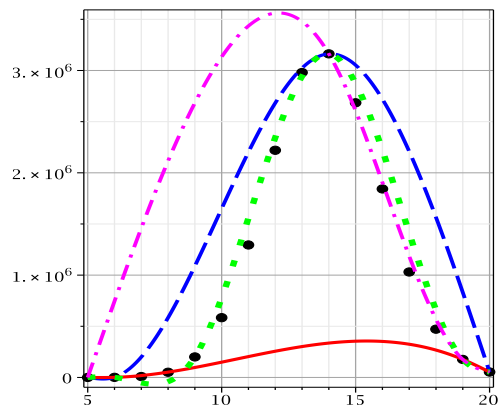}
\caption{Coefficients of the 5 by 5 hammock networks (black solid circles), coefficients obtained by our approximation (solid red line for $s=t=1$, dashed blue line for $s=1,t=6$, dash-dotted magenta line for $s=9,t=1$), and coefficients obtained by Spline approximations (dotted green line) when quasi-uniformly distributed $N_i$ are known (in this case $N_5,N_6,N_9,N_{14},N_{19},N_{20}$) .}
\label{fig:spline+cubic}
\end{figure}
\section{Conclusions}
An algorithm to approximate the reliability polynomials of two dual 2TNs is deduced and tested through simulations. The mathematical model leading to this algorithm is based on approximating the coefficients function of the reliability polynomial of a hammock network by a cubic spline function. The cubic spline is generated taking into account the complementarity with respect to the definite integration over the definition domain of the spline functions of coefficients of two dual hammock networks. The cubic spline scheme was chosen, using Bernstein and B\'{e}zier type approximation operators, to use their property of preserving the convexity and concavity of these operators. The input data is not sufficient for classically writing the approximant. Therefore, the mutual behaviour of the reliability polynomials of two dual hammock networks is used to generate a system of input constraints. These constraints are sufficient to produce the compatible system of equations that give the average value of the coefficients of two complementary reliability polynomials, but they are not sufficient for a refined estimation of the error. A convenient choice of the input constraints may refine the upper and lower bounds of the error. Numerical applications show that our algorithm produces results having better upper and lower bounds of the error, in Chebyshev norm, than in case of other approximation schemes from literature.

\medskip
\section*{Aknowledgement}
V-F. Dr\u{a}goi was supported by the European Union through the European Regional Development Fund under the Competitiveness Operational Program through the Project on Novel Bio-inspired Cellular Nano-Architectures under Grant POC-A1.1.4-E-2015 nr. 30/01.09. 2016.

\end{document}